\documentclass[11pt]{article}

\usepackage{amssymb}
\usepackage{amsmath}
\usepackage{color} 
\usepackage{fullpage}
\ifx\pdftexversion\undefined
\usepackage[dvips]{graphicx}
\else
  \usepackage[pdftex]{graphicx}
\fi

\usepackage[lined,boxed,commentsnumbered,linesnumbered]{algorithm2e}
\usepackage{array}
\usepackage{multirow}
\usepackage{hyperref}

\newtheorem{theorem}{Theorem}

\newtheorem{proposition}[theorem]{Proposition}

\newenvironment{proof}[1][Proof]{\noindent\textbf{#1.} }{\hfill $\Box$\\[2mm]} 

\newcounter{linenumber}

\def\L{\ensuremath{\mathcal{L}}}

\def\Nat{\ensuremath{\mathbb{N}}}

\def\argmin{\mathit{argmin}}

\def\ipart{\infty\text{-}\mathit{part}}
\def\part{\mathit{part}}

\def\countf{\mathit{countf}}
\def\lastf{\mathit{lastf}}

\newcommand{\correct}{\mathit{correct}}
\newcommand{\strcorrect}{\textit{str-correct}}

\newcommand{\true}{\mathit{true}}
\newcommand{\false}{\mathit{false}}

\newcommand{\remove}[1]{}

\newcommand{\ignore}[1]{}

\begin{document}
\bibliographystyle{abbrv}

\title{Strong Equivalence Relations for Iterated Models}

\author{
Zohir Bouzid$^1$
~~~~Eli Gafni$^2$~~~~Petr Kuznetsov$^1$\\
\\
\\
\large $^1$ T\'el\'ecom ParisTech\\
\large $^2$ UCLA
}

\date{}

\maketitle

\begin{abstract}

The Iterated Immediate Snapshot model  (IIS), due to its elegant
geometrical representation, 
has become standard for applying topological reasoning
to distributed computing.
Its modular structure makes it easier to analyze than 
the more realistic
(non-iterated) read-write Atomic-Snapshot memory model (AS).
It is known that AS and IIS are 
equivalent with respect to \emph{wait-free task} computability: a distributed task is solvable in AS if and only if
it solvable in IIS. 
We observe, however, that this equivalence is not sufficient in
order to explore solvability of tasks in \emph{sub-models} of AS (i.e. proper subsets of
its runs) or computability of \emph{long-lived} objects, and a
stronger equivalence relation is needed.  

In this paper, we consider \emph{adversarial} sub-models of AS and IIS
specified by the sets of processes that can be \emph{correct} in
a model run.
We show that AS and IIS are equivalent in a strong way:
a (possibly long-lived) object is implementable in AS under a given adversary 
if and only if it is implementable in IIS under the same adversary. 
Therefore, the computability of any object in shared memory under an
adversarial AS scheduler can be equivalently investigated in IIS.




\end{abstract}



\section{Introduction}
Iterated memory models (see a survey in~\cite{Raj10})
proved to be a convenient tool to investigate and
understand distributed computing.
%
In an iterated model, processes pass through
a series of disjoint communication-closed memories $M_1$, $M_2$,
$\ldots$. 
The most popular one is the \emph{Iterated Immediate Snapshot} model (IIS)~\cite{BG97}. 
Processes access the memories one by one, each time invoking the
\emph{immediate snapshot} operation~\cite{BG93a} that writes to the
memory and returns a snapshot of the memory contents.
Once memory $M_k$ is accessed, a process never comes back to it.
IIS has many advantages over the more realistic
(non-iterated) read-write Atomic-Snapshot memory model (AS)~\cite{AADGMS93}.
Its modular structure makes it considerably easier to analyze algorithms and prove their correctness.
Moreover, its nice geometrical representation \cite{Koz12, Lin10}
makes it suitable for  topological reasoning.
It is natural therefore to seek for a generic trasformation that would
map any problem in AS to an \emph{equivalent} problem in IIS.

It has been shown by Borowski and Gafni~\cite{BG97} that the complete
sets of runs of IIS and AS are, in a strict sense,
\emph{equivalent}:  a distributed task is (wait-free) solvable in AS if and only
if it is (wait-free) solvable in IIS.  
They established the result by presenting a \emph{forward simulation} that, in every AS run,
simulates an IIS run~\cite{BG93a}, and a \emph{backward simulation} that, in every IIS
run, simulates an AS run~\cite{BG97}. 
The equivalence turned out to be instrumental, e.g., in deriving
the impossibility of wait-free set agreement~\cite{BG93b,HS99}. More
generally, the equivalence enables the topological characterization of
task solvability in AS~\cite{HS99,HKR14}.  

However, in order to investigate computability of \emph{long-lived} objects or 
solvability of tasks in \emph{sub-models} of AS (i.e., proper subsets of its runs),  
this equivalence turns out to be insufficient.
The goal of this paper is to establish a stronger one using 
elaborate model simulations.

We focus on \emph{adversarial} sub-models of AS~\cite{DFGT11,Kuz12},
specified by sets of processes that can be \emph{correct} in
a model run. 
Note that the original AS model is described by the adversary 
consisting of \emph{all} non-empty sets of processes.  
Since the introduction of adversaries in~\cite{DFGT11}, the models have become popular
for investigating task computability~\cite{GK10,GK11,HR13}.
But how to define an IIS ``equivalent'' for an adversarial AS sub-model? 

In IIS, a correct yet ``slow'' process may be never noticed by other
processes: a process may go through infinitely many memories
$M_1,M_2,\ldots$ without appearing in the snapshots of any other process. 
Instead, we specify adversarial sub-models of IIS using the sets of
\emph{strongly correct} processes~\cite{RS12} (sometimes also referred
to as \emph{fast} processes~\cite{GKM14-podc}).
Informally, a process is strongly correct in an IIS run if it belongs to the largest set of
processes that ``see'' each other infinitely often in the run.
A topological characterization of task computability in sub-IIS models
has been recently derived~\cite{GKM14-podc}: given a task $T$ and an
IIS sub-model $M$, topological conditions for solving $T$ in $M$ are provided.
Is this characterization relevant for sub-AS models also or is it specific to IIS?


%

In this paper, we show that the answer is ``yes''. 
We show that sub-models of IIS and AS that are governed by the
same adversary are equivalent in a strong sense:
An object is implementable in AS under a given adversary if and only if it is implementable 
in IIS under the same adversary. This holds regardless of whether the
object is one-shot, like a distributed task, or long-lived, like
a queue or a counter.
To achieve this result, we present a two-way simulation protocol that provides 
an equivalent sub-IIS model for any sub-AS model
and which guarantees that the set of correct processes in an AS run coincides with the
set of strongly correct processes in the simulated IIS run, and vice
versa:


\vspace{2mm}
\noindent
$\bullet$ We propose an ``AS to IIS'' simulation 
%
which ensures that a correct (in AS) process is ``noticed''
infinitely often by other correct (in AS) processes in the simulated IIS run,
even if the process is much slower than the others.
To this goal,  we simulate IIS steps with the RAP
(Resolver Agreement Protocol)~\cite{GK11} and employ 
a ``fair'' simulation strategy---at each point, we first try to promote  
the most ``left behind'' process in the currently simulated run.
Even if the RAP-based simulation ``blocks'' because of a disagreement
between the simulators (unavoidable in asynchronous fault-prone
systems~\cite{FLP85}), we guarantee that the 
blocked process is eventually noticed by more advanced simulated
processes.

\vspace{2mm}
\noindent
$\bullet$ To obtain our ``IIS to AS' 
simulation, 
   we extend the multiple-shot IS simulations~\cite{BG93a} with a
   ``helping''  mechanism, reminiscent to the one employed in the atomic-snapshot
  simulation~\cite{AADGMS93}. 
   Here even if a process $i$ is not able to complete its simulated read, it may adopt
  the snapshot published by a concurrent process $j$, under the condition
  that $j$ has seen the most recent write of $i$. Since every move by
  a strongly correct process
  is eventually seen by every other strongly correct process, we derive the
  desired property that every strongly correct process makes progress in the
  simulated run. 

\vspace{2mm}

Equating the set of correct processes with the set of strongly correct processes in
the iterated simulated run is illuminating, because our algorithms 
provide an iterated equivalent to any adversarial
model~\cite{DFGT11,Kuz12}.
Our simulations also preserve the set of processes considered to be
participating in the original run,  
 which motivates the recent topological characterization of
task computability in sub-IIS models~\cite{GKM14-podc}.  
 
An important property of our simulation algorithms is that they are
model-independent, i.e., they deliver the promised guarantees without making any
assumptions on the model runs. In this sense, the algorithms are
\emph{wait-free}. 


\vspace{1mm}\noindent\textbf{Roadmap.}
Section~\ref{sec:related} relates our results to earlier work. 
Our model definitions, including the discussion of the AS and IIS
models, the definition of strongly connected 
processes in IIS, and the definition of a simulation, are given in Section~\ref{sec:model}.
Sections~\ref{sec:lf} and~\ref{sec:fl} present our two-way simulation. 

\section{Related work}
\label{sec:related}

The IIS model introduced by Borowsky and Gafni~\cite{BG97} has become standard in
topological reasoning about distributed computing~\cite{HS93,BG93b,BG97,HS99,HKR14}.
The IIS model is precisely captured by the
standard chromatic subdivision of the input
complex~\cite{Lin10,Koz12}, and thus enables intuitive and elegant
reasoning about its computability power, in particular, distinguishing
solvable and unsolvable.  
The IIS model is equivalent to the classical read-write model with
respect to (wait-free) task solvability~\cite{BG93a,BG97,GR10-opodis,RS12}. 

On the one hand, Borowsky and Gafni~\cite{BG93a} have shown that one
round of IIS can be implemented wait-free in AS, thus
establishing a wait-free simulation of multi-round IIS.
But the simulation only ensures that \emph{one} correct process appears as strongly correct
in the IIS run.  
Our algorithm ensures that \emph{every} correct
processes appear as strongly correct in the simulation.

On the other hand, IIS can simulate AS in the non-blocking manner, i.e., making sure that at least
one process that participates in infinitely many rounds of IIS manages
to simulate infinitely many steps of AS~\cite{BG97}. 
Later, Gafni and Rajsbaum~\cite{GR10-opodis} generalized the simulation
of~\cite{BG97} to $\L$-resilient adversaries~\cite{GK11}.
It guarantees that at least one set in $\L$ will appear correct in the
simulated execution.
Raynal and Stainer~\cite{RS12} 
presented an extension of the simulation in~\cite{BG97} and sketched  
a proof that the extension simulates a run in which each the set of
correct process in the simulated AS run is equal to the set of
strongly correct processes in the ``simulating'' IIS run.
In this paper, we propose an algorithm that achieves this property
using the idea of the original atomic-snapshot
implementation by Afek et al.~\cite{AADGMS93}, which we believe to be more intuitive and simpler to understand.


The relations between different simulation protocols are summarized in
the following table: 

\vspace{1mm}
\hspace{-5mm}
\begin{tabular}{|l|c|c|} 
   \hline
    & correct(AS) $\subseteq$ str-correct(IIS)? & str-correct(IIS) $\subseteq$ correct(AS)?  \\
   \hline
    \hline
    \textbf{From AS to IIS} & \multicolumn{2}{c|}{} \\
    \hline
    \small{Borowsky and Gafni~\cite{BG93a}} & \footnotesize{$\exists p \in \text{correct(AS)}: p \in \text{str-correct(IIS)}$}  & \checkmark \\
    
    \textbf{This paper} & \checkmark & \checkmark \\
    \hline
   \hline
   \textbf{From IIS to AS} & \multicolumn{2}{c|}{} \\
   \hline
   \small{Borowsky and Gafni~\cite{BG97}} & \checkmark & \footnotesize{$\exists p \in \text{str-correct(IIS)}: p \in \text{correct(AS)}$}  \\
   
   \small{Gafni and Rajsbaum~\cite{GR10-opodis}} & \checkmark & \footnotesize{$\exists X \subseteq \text{str-correct(IIS)}:  X \subseteq \text{correct(AS)}$ } \\
  
    \small{Raynal and Stainer~\cite{RS12}} & \checkmark & \checkmark \\
    
   \textbf{This paper} & \textbf{\checkmark} & \checkmark \\
    \hline

\end{tabular}
\vspace{1mm}
\normalsize

An informal definition of a strongly correct process in IIS was 
proposed by Gafni in~\cite{Gaf98-iis} and formally stated by Raynal and Stainer
in~\cite{RS12,RS13}.
The equivalence between adversarial restrictions of
AS and IIS we establish in this paper motivated formulating 
a generalized topological characterization of task computability in sub-IIS~\cite{GKM14-podc}.  

Our AS-to-IIS simulation presented in Section~\ref{sec:lf} offers a
novel use of the Resolver Agreement Protocol (RAP) proposed
in~\cite{GK11}, where a set of simulators try to maintain the balance
between the simulated processes by promoting the ``most
behind'' process that is not ``blocked.'' 
Our IIS-to-AS simulation presented in Section~\ref{sec:fl} is based on
the non-blocking simulation of~\cite{BG97}, with the helping mechanism
similar to the one used in the original atomic snapshot
construction~\cite{AADGMS93}.    

Herlihy and Rajsbaum~\cite{HR12-sim} considered the problem of 
simulating task solutions in a variety of models, but their
results only concern colorless tasks, which boils down to a very
restricted notion of simulation. 
%
Rajsbaum et al.~\cite{RRT08} introduced the Iterated \emph{Restricted}
Immediate Snapshot (\emph{IRIS}) framework, where the restriction is defined via a specific
failure detector on the per-round basis.

\section{Definitions}
\label{sec:model}

In this section, we recall how the standard read-write and
IIS models are defined, discuss
the notion of a strongly correct process in the IIS model, and explain what we
mean by simulating one model in another.

\vspace{1mm}\noindent\textbf{Standard shared-memory model.}
We consider a standard atomic-snapshot model (\textit{AS}) in which a
collection $\Pi=\{1,\ldots,n\}$ 
of processes communicate via atomically updating their distinct registers in the memory and taking
atomic snapshots of the memory contents.
AS is equivalent to the standard read-write shared-memory model~\cite{AADGMS93}. 
Without loss of generality, we assume that every process writes its
input value in the first step and then alternates taking snapshots with
updating its register with the result of it latest snapshot. 
This is known as a \emph{full-information} protocol.
We say that a process \emph{participates} in a run $E$ if it performs at
least one update operation.
Let $\part(E)$ denote the set of participating processes in $E$. 
A process $i$ is \emph{correct} (or\emph{live}) in $E$ if $i$ takes infinitely many
steps in $E$. Let $\correct(E)$ denote the set of live processes in $E$. 
 
\vspace{1mm}\noindent\textbf{IIS model and strongly correct processes.} 
%
In the IIS memory model, each process is supposed to go through a
series of independent memories $M_1$, $M_2$, $\ldots$.
Each memory is accessed by a process with a single \emph{immediate snapshot}
operation~\cite{BG93a}.

A run $E$ in IIS is a sequence of non-empty sets
of processes $S_1 \supseteq S_2 \supseteq \dots$, with each $S_r
\subseteq \{1, \dots, n\}$ consisting of those processes that
participate in the $r$th iteration of immediate snapshot
(IS). Furthermore, each $S_k$ is equipped with an ordered partition:
$S_r = S_r^1 \cup \dots \cup S_r^{n_r}$ (for some $n_r \leq n$),
corresponding to the order in which processes are invoked in the
respective IS. 
 
Fix a run $E = S_1,S_2, \ldots$. The processes $i \in S_1$ are called {\em participating}.  
If $j$ appears in all the sets $S_k$, we say that $j$ is
{\em infinitely participating} in $E$. 
The sets of participating and infinitely participating processes in a run $E$ are denoted $\part(E)$ and $\ipart(E)$, respectively.

If $i\in S_r$ ($i$ \emph{participates in round $r$}), let $V_{ir}$ denote the set of processes appearing  in $i$'s $r$-th \emph{snapshot}
in $E$, defined as the union of all sets in the partition of $S_r$
preceding and including $S_r^{m}\subseteq S_r$ such that $i\in S_r^{m}$: $V_{ir}= S_r^1 \cup \dots \cup S_r^m$.  
It is immediate that for all processes $i,j$ and rounds $r$, such that $i$ and $j$
participate in $r$, the following properties
are satisfied~\cite{BG93a}:
(self-inclusion) $i \in V_{ir}$;
(containment) $V_{ir}\subseteq V_{jr}$ $\vee$ $V_{jr}\subseteq V_{ir}$; and 
(immediacy) $i \in V_{jr}$ $\Rightarrow$ $V_{ir}\subseteq V_{jr}$.



Our definitions can be interpreted operationally as follows.
$S_r$ is the set of processes accessing memory $M_r$, and each $S_r^j$ is the set of processes
obtaining the same \emph{snapshot} after accessing $M_r$.
Recall that in IS, the view of a process $i\in S_r^j$ is defined by the values
written by the processes in $S_r^1 \cup \dots \cup S_r^j$.     




%
It is convenient then to define, for each round $r$ of $E$, a directed graph $G_E^r$
with processes that participate in $r$ as nodes and a directed edge from $i$ to
$j$ if $j \in V_{ir}$.
$G^{(r)}_E$ is the union of the graphs $G_r$,$G_{r+1}, \ldots$

We say that \emph{process $i$ is aware of round $r$ of process $j$} 
in an IIS execution $E$ 
if there exists a path from $i$ to $j$ in $G^{(r)}_E$.

The \emph{participating set of a process $i$ in a run $E$},  denoted by $\part(E,i)$ (or mayby $aware(i,E)$?), is the set of
processes that $i$ is aware of their first round.

A process is \emph{strongly correct} (or
\emph{fast}~\cite{GKM14-podc}) in $E$ if  in every round
by every process in $\ipart(E)$. Let $\strcorrect(E)$ denote the set of
strongly correct processes in $E$.
Intuitively, $\strcorrect(E)$ is the largest set of processes that ``see''
each other (appear in each other's views) infinitely often in $E$.
Formally, denote by $G^{*}_E$ the graph limit $\lim\limits_{r \to \infty} G^{(r)}_E$.
That is, $i$ is a vertex of $G^{*}_E$ if it is in $\ipart(E)$
and $(i, j)$ is an edge of $G^{*}_E$ if
$E$ contains infinitely many
rounds $r$ such that $j\in V_{ir}$, i.e., $i$ is aware of  infinitely
many rounds of $j$.

Let $SC(E)$ be
the set of processes in the strongly connected component of
$G^{*}_E$.
By the containment property of IIS snapshots,
in every round $r$, either 
$i\in V_{jr}$ or $j\in V_{ir}$. Hence, for all $i,j\in\ipart(E)$,
we are guaranteed that $G^{*}_E$ contains 
at least one of the edges $(i, j)$ and $(j, i)$.
Therefore, $G^{*}_E$ has a single sink. 
In the following, we use some properties of strongly correct processes:


\begin{proposition}
\label{prop:fast-graph-fast}
For all $E$ in IIS, $i \in \strcorrect(E)$
iff 
there exists $r_0$, such that for all $r\geq r_0$, 
$G^{(r)}_E$ contains a path between every process in $V_{ir}$ and $i$.
\end{proposition}

\begin{proof}
\begin{itemize}
\item[$\Rightarrow$]
Let $i\in \strcorrect(E)$. 
Since $\strcorrect(E) \subseteq \ipart(E)$, $i$ belongs also to $\ipart(E)$.
Take $r_0$ as the first round such that for all $r \geq r_0: V_{ir} \subseteq \ipart(E)$.
$r_0$ is well defined since the processes not belonging to $\ipart(E)$ can appear only finitely often in the snapshots of $i$.

Since $i$ is strongly correct, for all $r \geq r_0$, $G^{(r)}_E$ contains a path between every process $j \in \ipart(E)$ and $i$.
But as $r \geq r_0$, $V_{ir} \subseteq \ipart(E)$.
Hence, $G^{(r)}_E$ contains a path between every process of $V_{ir}$ and $i$.


\item[$\Leftarrow$]
Let $i$ be a process and $r_0$ a round such that for all $r\geq r_0$,
$G^{(r)}_E$ contains a path between every process in $V_{ir}$ and $i$.
We need to prove $i \in \strcorrect(E)$.

Note that the containment property of IS snapshots guarantees that 
every process $j \in \ipart(E) \setminus V_{ir}$
obtains a snapshot that contains $V_{ir}$. 
That is, $(j, i) \in G^{r}_E$ and hence $(j, i)\in G^{(r)}_E$.

Thus, we conclude that $G^{(r)}_E$ contains a path between every process in $\ipart(E)$ and $i$.

Since there are infinitely many such rounds $r$ and the number of possible paths is bounded, it follows
that $G^{*}_E$ contains a path between every process in $\ipart(E)$ and $i$.
Consequently, $i \in \strcorrect(E)$.
\end{itemize}
\end{proof}


\vspace{1mm}\noindent\textbf{Model simulations.}
In this paper we focus on models in which every process in a set
$1,\ldots,n$ alternates
writes with taking snapshots of (iterated or non-iterated) memory,
using the result of its latest snapshot (or its input value
initially) as the value to write.
Notice that the updates do not return any meaningful response, just
an indication that the operation is complete.
Thus, the evolution of the snapshot of a process $i$ in a run $E$ of such a
model is characterized by the sequence $V_{i,1}^E,V_{i,2}^E,\ldots...$
of the snapshots it takes in $E$. 

By simulation of a run of a model $B$ in another model $A$, we
naturally mean a distributed algorithm that in every run of $A$
outputs at every process a sequence of snapshots so that all these
sequences are consistent with some run of $B$ and, moreover, reflect the inputs
of $A$.  The latter intuitively filters out any ``fake'' simulation that
produces a run of $B$ that has nothing to do with the original run of $A$. 

Formally, in every run $E$ of $A$, a simulation $\textit{Sim}_{A,B}$ outputs, at
every simulator $i\in \{1,\ldots,n\}$ a (finite or infinite) sequence of snapshot values
$U_{i,1},U_{i,2},\ldots...$.       
There exists a run $E'$ of $B$ such that:
\begin{itemize}
\item For all $i$, $V_{i,1}^{E'},V_{i,2}^{E'},\ldots$ is exactly $U_{i,1},U_{i,2},\ldots...$;  
\item for every $i\in\correct(E)$ (resp., $\strcorrect(E)$ if $A$ is an IIS model), $\part(E,i)=\part(E',i)$.  
\end{itemize}
For the sake of brevity, we assume that in the simulated algorithm,
as its local state, each process $i$ simply maintains a vector storing the number of
snapshots collected by every other process it is aware of so far. 
The process writes the vector as its current state in write operation.
Each time a new snapshot is taken, the process updates its vector and 
simply increments its number of steps in it.
Initially, the vector of process $i$ stores $1$ in position $i$ and
$0$ at every other position.    
The reader can easily convince herself that this simplification does
not bring a loss of generality, i.e., provided a simulation for
such an algorithm, we can derive a simulation for the full-information
algorithm.

\section{From AS to IIS: resolving and bringing to the front}
\label{sec:lf}

The goal of this section is to provide an algorithm
$\textit{AS}\rightarrow\textit{IIS}$ that simulates an
execution of an IIS model where the set of
processes that appear strongly correct coincides with the set of correct processes
(Algorithm~\ref{alg:lf}).

\vspace{1mm}\noindent\textbf{Overview.} 
For each iteration of the IIS model, the processes use the original IS
implementation~\cite{BG93a}.
To ensure fairness of the simulation, each 
process tries to advance the process that is currently the
most behind.


Recall that the IS construction~\cite{BG93a} involves $n$ recursive
levels, $n$ down to $1$, where at each level $\ell$,  every process
registers its participation and then takes an atomic snapshot.
If the size of the snapshot is less than $\ell$, then the process
recursively proceeds to level $\ell-1$, otherwise it returns the
snapshot as its output in the IS simulation.
Since at most $n$ processes start at level $n$ and at least one
process (the one that writes the last) drops the simulation at each
level, at most $\ell$ processes can reach level $\ell$.
 
In our $\textit{AS}\rightarrow\textit{IIS}$ algorithm, in order to promote the next step of a given
process, the simulators use an
\emph{agreement protocol}~\cite{BG93b,Kuz13} for each level of the IS simulation~\cite{BG93a}.
More precisely, to simulate the atomic snapshot taken by the process in
level $\ell$, 
the simulator takes an atomic snapshot itself to compute
the set of other simulated processes that also reached level $\ell$.
If the cardinality of the set is exactly $\ell$, then the simulator
proposes $1$ to the agreement algorithm. Otherwise, it proposes $0$.    
If the agreement protocol returns $1$, then the simulated process completes the IS
iteration by outputting the set of $\ell$ processes in level $\ell$.
If the agreement protocol returns $0$, the process gets down to level
$\ell-1$ in the current IS iteration.
  
To make sure that the simulation is safe, i.e., the simulators indeed agree
on the outcome of the simulated step, we use the recently proposed  \emph{Resolver Agreement
  Protocol (RAP)}~\cite{GK11}.
This protocol guarantees agreement (no two processes output different
values) and validity (every output value was previously proposed).
Moreover,  if all proposed values are the same, then the algorithm
terminates. 
This feature is implemented using the \emph{commit-adopt} (\emph{CA})
algorithm~\cite{Gaf98}. 
Otherwise, if two different values are proposed,  
the agreement algorithm may block.
The blocked state can be \emph{resolved} by the simulated process itself: 
the simulated process writes the value it adopted from CA in a
dedicated register so that every correct process would eventually read
the value and terminate.   

Formally, the RAP 
exports one operation $\textit{propose}(v)$, $v\in \{0,1\}$ 
that returns a value in $\{0,1,\bot\}$
and is associated
with a unique \emph{resolver} process.
The following guarantees are provided:    
(i) Every returned non-$\bot$ value is a proposed value; 
(ii) If all processes propose the same input value, then 
  no process returns $\bot$;
(iii) The resolver never returns $\bot$;	 
(iv) No two different non-$\bot$ values are returned.

\vspace{1mm}\noindent\textbf{Operation.} 
Algorithm~\ref{alg:lf} operates as follows.
Every process maintains a shared vector $R[i]$, written by $i$ and
read by all, that stores $i$'s perspective on the current simulation. 
In particular, the sequence of iterations $r$ and levels $\ell$ that a process $j$ has passed
through, \emph{as witnessed by $i$}, is stored in $R[i,j]$. 

After taking a snapshot $S$ in line~\ref{line:snapshot} of the current
simulated state, the simulator $i$ first checks if the simulated
process $i$ is \emph{blocked} (line \ref{line:ifblocked}).
A process $p$ is considered blocked if for every $S[j,p]$ that 
contains $(d,r,\ell)$ with 
$(r,\ell)=\textit{round-level}(p,S)$, we have $d=\textit{blocked}$. 
If simulated process $i$ is blocked in the simulation, simulator $i$
retrieves the round-level $(r,\ell)$ at which it is blocked (line \ref{line:retrieveLevel})
and participates in $\textit{RAP}_{i,r,\ell}$.
We assign $i$ to be the \emph{resolver} of each RAP instance
$\textit{RAP}_{i,r,\ell}$, and thus the instance returns a non-$\bot$
which ``unblocks'' simulated process $i$.

If simulated process $i$ is not blocked, simulator
$i$ checks if some process $j$ has completed a \emph{new} (not
considered by $i$ in previous rounds of the simulation) round $r_j$,
such that all processes in $V_{jr_j}$ are aware of round $r_j$ of $j$ (line \ref{line:condFreeze}). 
Every such process $j$ is then \emph{frozen} by $i$, i.e., $j$ is put
on hold and not simulated until simulator $j$ 
performs a ``physical'' step (in lines~\ref{line:move1} or~\ref{line:move2}).

\remove{
				After taking a snapshot $S$ in line~\ref{line:snapshot} of the current
				simulated state, the simulator first computes the set of 
				``most advanced'' processes, i.e., the processes that, according to
				the outcome of the snapshot, complete the latest round with the
				smallest snapshot (line~\ref{line:fastest}.  
				Here $\textit{round-level}(p,S)$ denotes the maximal round-level
				reached by $p$ in $S$, i.e., the maximal value of the last element in
				$S[*,p]$.
				If the processes in the resulting \emph{front} set are about to finish
				round $r$, 
				then every such process is 
				``frozen'', i.e., it is put on hold and not simulated until it performs a physical step (in
				lines~\ref{line:move1} or~\ref{line:move2}). 
}

In the set of remaining processes, the simulator chooses  
the ``slowest'' \emph{non-blocked} and \emph{non-frozen} process (line~\ref{line:slowest}).  
To make sure that the notion of the slowest process is well-defined,
we introduce a total order on the tuples $(i, r, \ell)$, $i\in\Pi$, $r
\in \mathbb{N}$, $\ell \in \mathbb{N}_n$ as follows.
We say that $(i, r_i,\ell_i)<(j, r_j,\ell_j)$
if
$(r_i<r_j) \vee ((r_i=r_j) \wedge (\ell_i>\ell_j)) \vee
((r_i=r_j) \wedge (\ell_i=\ell_j) \wedge ((i+r_i) \mod n<(j+r_i) \mod n))$. 
This way $\textit{argmin}$ in line~\ref{line:slowest} returns a single process, 
ties are broken by choosing the process associated with the current
iteration (the association is done in round-robin).

The slowest process $p$ currently observed (by $i$) in round-level $(r,\ell)$
is then simulated using $p$'s next instance
of RAP, $\textit{RAP}_{p,r,\ell}$, which accepts either $1$ (exactly
$\ell$ processes have appeared on round-level $(r,\ell)$ in $S$) or 
$0$ (otherwise). If $\textit{RAP}_{p,r,\ell}$ returns $1$, $p$
outputs the set of $\ell$ processes in $(r,\ell)$ as its snapshot
in
round $r$, denoted $V_{pr}$, and then $p$ is promoted to round $r+1$ (lines~\ref{line:nextround}
and~\ref{line:output}).
If  $\textit{RAP}_{p,r,\ell}$ returns $0$, $p$
is promoted to level $\ell-1$ of the same round $r$
(line~\ref{line:nextlevel}).
Otherwise, if $\textit{RAP}_{p,r,\ell}$ is blocked, we mark the status
of $i$ as \emph{blocked} in $(r,\ell)$ (line~\ref{line:blocked}).  

\vspace{1mm}\noindent\textbf{Correctness intuition.}
Our algorithm tries to always promote the process that is the
``most left behind'' process to the front of the simulation, unless a process gets
\emph{blocked} or \emph{frozen}.

A process $i$ is blocked if two simulators
proposed two different values to some $\textit{RAP}_{i,r,\ell}$ , i.e.,
one simulator finds exactly $\ell$ processes in $(r,\ell)$ and,
thus, believes that $i$ should complete round $r$ by outputting the
$\ell$ processes, and the other
found strictly less processes  in $(r,\ell)$ and thus believes that $i$ should go one
level down in round $r$ and output a smaller snapshot.
A process is frozen if it produced a new snapshot in a round $r$ and all the 
processes appearing in this snapshot became aware of it. The intuition here is that, according to Proposition \ref{prop:fast-graph-fast},
strongly correct processes are frozen infinitely often.
Therefore, only a correct simulator $i$ may appear as strongly correct in the simulated
run: otherwise the corresponding simulated process $i$ would get
frozen after $i$ crashes and stay frozen forever (only $i$ can
``unfreeze'' itself in the simulation).
 
Intuitively, a process $i$ is blocked because another process appeared at
its round-level $(r,\ell)$ and two simulators disagreed whether the
other process was there or not: one simulator finds exactly $\ell$
processes at the level and the other strictly less processes. 
The last such process $p$ will now be considered the
slowest process in the simulation and, thus, will be chosen to be
promoted in line~\ref{line:slowest} by any other simulator. Note that $p$ cannot be blocked
in $(r,\ell)$,
because every simulator that found $p$ in $(r,\ell)$ will also find
exactly $\ell$ processes in $(r,\ell)$. This is because $p$ is the
last process to reach $(r,\ell)$. 
Moreover, $p$ completes iteration $r$ having $i$ in its snapshot: since
$p$ completes $r$ in level $\ell$ reached by $i$,  $p$ \emph{sees}
$i$ in round $r$ in the simulated run. 
By repeating this reasoning inductively, 
even though $i$ is blocked, another process $p$ carries
this information to the ``front'' of the simulation, thus
making sure that every other simulated process will eventually be
aware of round $r$ of $i$.
Process $i$ unblocks itself by completing its own  $\textit{RAP}_{i,r,\ell}$
an thus providing it with a non-$\bot$ output.


\SetKw{KwShared}{Shared:}
\SetKw{KwLocal}{Local:}
\begin{algorithm}[t]
{\scriptsize
\KwShared{$R[1],\ldots,R[n] := [\bot,\ldots,\bot],\ldots,[\bot,\ldots,\bot]$\;}
\KwShared{$\textit{Counter}_1,\ldots,\textit{Counter}_n  := 0,\ldots,0$\;}
\KwLocal{$\countf[1,\ldots,n]  :=  [0,\ldots,0]$} \tcp*[r]{ counters for
  ``frozen'' processes}	    
\KwLocal{$\lastf[1,\ldots,n]  :=  [0,\ldots,0]$}  \tcp*[r]{last rounds
  in which processes were ``frozen''} 
\vspace{3mm}
       $R[i,i] := (runm0,n)$ \label{line:init} \tcp*[r]{start with highest level of
  the first iteration}
       \While{\textit{true}}{
              $\textit{Counter}_i ++$\label{line:move1}\; 
              $S :=$  snapshot of
              $R[1],\ldots,R[n]$ \label{line:snapshot}\;
 
           \eIf{$i$ is blocked in $S$}
              {          
              			$p := i$ \label{line:ifblocked}\;
              }
              {
			              \For{each $j \in \Pi$}{
				              $x :=$ the largest round such that $V_{jx}$ are aware of round $x$ of ${j}$ (in $S$) \label{line:condFreeze}\;

                                        \label{line:condFreeze1}\If{$x > \lastf[j]$} 
				              {
					              $\lastf[j] := x$ \;
					              $\countf[j] := \textit{Counter}_j$ \tcp*[r]{freeze $j$}
					              
					          }
			             }
			              \Repeat{$\textit{cands}\neq\emptyset$}{
			              $\textit{cands} := \{ j ~|~ j \text{ is not blocked and }
			              \textit{Counter}_j > \countf[j]
                                      \}$\label{line:cands}  \tcp*[r]{
                                      ignoring non-participants}	    
			              $\textit{Counter}_i ++$\label{line:move2}\;
			              }  
			              $p := \argmin _{j\in\textit{cands}} (\textit{round-level}(j,S), (j+\textit{round}(j,S))\mod n)$ \label{line:slowest}\tcp*[r]{choose the ``most-behind'' process}	                     }        

			   $(r,\ell) := \textit{round-level}(p,S)$
			   \tcp*[r]{compute current round and level of $p$}\label{line:retrieveLevel}
			   $U := \{j ~|~  (*,r,\ell)\in S[*,j]\}$ \tcp*[r]{all
			   processes reached $(r,\ell)$} 
			   $v := \textit{RAP}_{p,r,\ell}(|U|=\ell)$ \tcp*[r]{the result of next step of $p$}\label{line:doRAP}
	              
              \eIf {$v=\true$} { 
                $R[i,p] := S[i,p]\cdot (run,r+1,n)$ \label{line:nextround}\tcp*[r]{$p$ completes
                  round $r$}  
                  $V_{pr} := U$\label{line:output} \tcp*[r]{output
                the snapshot of $p$ in round $r$}
              }
              {\eIf{$v=\false$}{$R[i,p] := S[i,p]\cdot (run,r,\ell-1)$ \label{line:nextlevel}
                  \tcp*[r]{$p$ proceeds to $(r,\ell-1)$} }
                 { $R[i,p] :=  S[i,p] \cdot (blocked,r,\ell)$
                    \label{line:blocked} \tcp*[r]{$p$ blocks in $(r,\ell)$}}
              } 

        }
}
\caption{The $\textit{AS}\rightarrow\textit{IIS}$ simulation algorithm: code for process $i$.}
\label{alg:lf}
\end{algorithm}

Thus, intuitively, a correct process
$i$ either gets blocked infinitely often or gets frozen infinitely
often. In both cases, $i$ is ``seen'' infinitely often by other correct
processes. 
Moreover, a faulty process is eventually either (i) gets faulty of
frozen forever, or (ii) becomes invisible to the remaining processes
in the simulated run. In both cases, the faulty process does not
appear strongly correct in the simulation.  
Thus:

\begin{theorem}
\label{th:lf}
 Algorithm~\ref{alg:lf} provides a simulation of the IIS model in the
AS model such that, for each run $E$, the
 simulated run $E'$ satisfies (1) $\correct(E)=\strcorrect(E')$, and 
(2) $\forall i\in\correct(E)$: $\part(E)=\part(E',i)$.
\end{theorem}
\begin{proof}
Take any run $E$ of Algorithm~\ref{alg:lf}. 
Recall that the simulated run $E'$ is defined as a collection of all sets
$V_{ir}$, $i=1,\ldots,n$, $r\in\Nat$ produced in $E$. 
By the correctness of the IS simulation~\cite{BG93a} and the use of the RAP
agreement protocol~\cite{GK11} for each atomic snapshot taken in the simulation
of~\cite{BG93a}, we conclude that for all $r$, all sets $V_{ir}$ satisfy
containment, self-inclusion and immediacy (defined in Section~\ref{sec:model}).  
Notice that by the algorithm, every correct process $i$ produces a
snapshot $V_{ir}$ in every iteration $r$.  


\vspace{1mm}\noindent\textbf{Every strongly correct process is correct.}
Assume for the sake of contradiction that $i \in \strcorrect(E')$ but $i \not\in live(E)$.
Define $r_0$ to be the first simulated round of $E'$  such that in all
$r\geq r_0$, 
(i) only processes of $\ipart(E')$ are simulated and (ii) $V_{ir}$ contain only strongly correct processes.
$r_0$ is well defined since the processes that appear infinitely often in the snapshots of strongly correct processes
are necessarily also strongly correct.

Take a round $r\geq r_0$ where $V_{ir}$ is simulated after the crash 
of $i$ in $E$ (recall that $i \not\in \correct(E)$).
Since $i$ is strongly correct, all the processes in $\ipart(E')$ (including
$V_{ir}$) will eventually be aware of round $r$ of $i$.
But the fact that the processes in $V_{ir}$ are strongly correct means that the
processes of $\ipart(E')$ are aware of infinitely many of their rounds.
Therefore, every process in $\ipart(E')$ eventually knows that 
the processes in $V_{ir}$ were aware of a round of $i$.
Hence, $i$ will be frozen by all of them. 
But since it has already crashed in $E$, it will never be unfrozen and 
cannot be simulated after $r$, contradiction.
Consequently, $i \in \correct(E)$ and $\strcorrect(E') \subseteq \correct(E)$.

\vspace{1mm}\noindent\textbf{Every correct process is strongly correct.} 
By Proposition \ref{prop:fast-graph-fast}, 
there exists a round $r_0$ such that for all $r \geq r_0$, the processes of $V_{ir}$
are aware of round $r$ of $i$  iff $i$ is strongly correct.
Hence, if a process is not strongly correct, 
the condition of line \ref{line:condFreeze1}) can apply to it only 
a finite number of times.
Thus, there exists a round $r_0^\prime$ 
such that every process that is frozen after 
it reaches $r_0^\prime$ is necessarily strongly correct.



Now we show that  $\correct(E)\subseteq\strcorrect(E')$.
Suppose not, i.e., 
there are processes $i,j\in \correct(E)$ and a round 
$r\geq r_0^\prime$
such that $i$ is never aware of round $r$ of $j$ in $E'$.
Since $r\geq r_0^\prime$, $j$ cannot be frozen by $i$.
Let $r_i$ be the round of process $i$ at the
moment when $j$ completes round $r$, i.e., outputs
$V_{rj}$ (line~\ref{line:output}).
 
Take $r'$ to be the first round greater than $r_i$, such that
$(i+r')\mod n+1=n$, i.e., $r'$ has the lowest priority in round $r'$.
Thus, before $i$ is simulated at some level $\ell'$ of $r'$, 
any other process that is not frozen or blocked 
must have competed its simulation of round $r'$ or reached level lower
than $\ell'$.

Let $\ell'$ be the level at which $i$ obtains its snapshot in $r'$
and let $m$ be some simulator that simulated $V_{ir'}$.

%

We observe first that $(r,\ell)<(r',\ell')$: otherwise, $i$ will
eventually reach level $\ell''\geq\ell'$ of round $r$, find exactly
$\ell''$ processes (including $j$) at that level, and output its
snapshot $V_{ir}$ such that $j\in V_{ir}$---a contradiction with the
assumption that $i$ is never aware of round $r$ of $j$.

Consider the time after $i$ reaches $(r', \ell')$ and before it obtains the snapshot $V_{ir'}$.
By the algorithm, the simulator $m$ must choose the slowest
non-blocked and non-frozen process to simulate. 
Suppose that $j$ is never observed blocked by $m$ after $i$ reaches $(r',\ell')$.
Since $j$ cannot be frozen by $m$ after $r_0^\prime$, the algorithm guarantees that 
eventually, 
$m$ would bring $j$ to level
$(r',\ell')$ and, thus, simulates a snapshot $V_{ir'}$ such that $j\in V_{ir'}$---a
contradiction.

Now suppose that $m$ observes $j$ as blocked in round $r$ or
later. Without loss of generality, suppose that $j$ is observed as
blocked by $m$ in round $r$. (Indeed, if $i$ is never aware of round
$r$ of  $j$ in, it is never aware of any later round of $j$.)

We claim that at the moment the first simulator took its snapshot on behalf of $j$ for round $r$ (in
line~\ref{line:snapshot}), there was another blocked process $k$
reached $(r,\ell)$ that later was observed as resolved by 
another simulator.
Indeed, the only reason for $j$ to block in $\textit{RAP}_{j,r,\ell}$
is that there is another simulator proposing a conflicting  set of
processes that have been observed to reach $(r,\ell)$.
Moreover, by the algorithm, since the simulators proposed different values to
$\textit{RAP}_{j,r,\ell}$ 
one of these sets contains exactly $\ell$ processes and
the other contains strictly less. 
Consider any process in the difference between these two snapshots
(the atomic snapshots taken in line~\ref{line:snapshot} are related by
containment~\cite{AADGMS93}).  
Every such process was considered blocked by one of the simulators at the moment it
took its snapshot in line~\ref{line:snapshot}, otherwise it would
appear in all obtained snapshots or would be chosen to be simulated
as slower process.
For the last such process $s$ to reach level $(r,\ell)$,
$\textit{RAP}_{s,r,\ell}$ cannot get blocked, because all simulators
will propose exactly $\ell$  processes that
reached $(r,\ell)$. 
Thus, $s$ obtains $V_{sr}$ such
that $j\in V_{sr}$ and enters round $r+1$. 

By our assumption, $(r+1,n)<(r',\ell')$ and $i$ is never aware of round $r+1$ of $s$.
Therefore, $s$ is not strongly correct and cannot be frozen as $r+1 \geq r_0^\prime$.
Moreover, $s$ does not block in round $r$, thus $m$ should eventually try simulating $s$
in round $r+1$. 
By repeating the argument inductively, we locate a process $t$ that
reaches round $r+2$ and is aware of  round $r+1$ of $s$.

Eventually,  some process that is aware of round $r$ of $j$ will reach 
$(r', \ell')$, and thus will appear in $V_{ir'}$.
Therefore, $i$ is aware of round $r$ of  $j$---a contradiction.

Finally, since every process starts the algorithm by registering its
participation at level
$(0,n)$ (line~\ref{line:init}),  
the set of participating processes in $E$ is automatically the
participating set for every correct process in $E'$.
\end{proof}

\section{From IIS to AS: identical snapshots and helping}
\label{sec:fl}

We now describe our $\textit{IIS}\rightarrow\textit{AS}$ algorithm that, in any run of the IIS model,
simulates a run of the AS model in which every process
alternates updates with atomic snapshots~\cite{AADGMS93}.

As a basis, we take the non-blocking simulation proposed by Borowsky and Gafni~\cite{BG97}.
In this algorithm, each process $i$ maintains a local \emph{counter} vector $C_i[1,\ldots,n]$
where each $C_i[j]$ stores the number of simulated snapshots of $j$ \emph{as
currently witnessed by $i$}. 
To simulate a snapshot operation, process $i$ accesses the
iterated memories, writing its counter vector $C_i$, taking a snapshot
of counter vectors of other processes, and updating each position 
$C_i[k]$ with the maximal  value of $C_j[k]$ across all counter
vectors read in the iteration.
In each iteration $r$ of the IIS memory, this is expressed as a single
$\textit{WriteRead}_r(C_i)$ operation the outcome of which satisfies the
self-inclusion, containment, and immediacy properties specified in Section~\ref{sec:model}.
If all these vectors are identical, $i$ outputs the vector as the
result of its next snapshot operation.     
Initially and each time a process $i$ completes its next snapshot operation, it
simulates an update operation by incrementing $C_i[i]$. 

We first observe that the original simulation of the
AS model proposed in~\cite{BG97} is, in the worst case, only
non-blocking. Indeed, it admits runs in which some strongly correct process is never able
to complete its snapshot operation, even though ``noticed''
infinitely often.     
Consider, for example, the following IIS run: $[\{1\}\{2,3\}]$, $[\{3\},\{1,2\}]$,
$[\{1\}\{2,3\}]$, $\ldots$, i.e., all the three processes are strongly correct and
in every iteration, one of the processes in $\{1,3\}$ only sees
itself and, thus, completes its new snapshot.
Thus, in every round one of the processes in $\{1,3\}$ outputs a new
snapshot, while the remaining process $2$ sees two different vectors 
and thus does not complete its simulated snapshot. 
As a result, process $2$ never manages to completes its first snapshot in the
simulated AS run, even though it is strongly correct!

To fix this issue, we equip the algorithm of~\cite{BG97} with a helping
mechanism, similar to the helping mechanism proposed in the atomic
snapshot simulation in~\cite{AADGMS93}. 
In addition to its counter vector, in each iteration of our
Algorithm~\ref{alg:fl}, 
a process also writes the
result of its last snapshot: $\textit{WriteRead}_r(C_i)$ (line~\ref{line:is}). 
Now a process $i$ outputs a new snapshot not only if it sees that
everybody agrees on the clock vector, but also if another process
produces a snapshot containing  $i$'s latest counter value.  

\begin{algorithm}[t]
{\scriptsize
%
      $C_i[1,\ldots,n] := [0,\ldots,0]$; $C_i[i] := 1$; $r := 0$; $SI_i := [0,\ldots,0]$\;
       \While{\textit{true}}{
            $r++$\; 
         $S :=  \textit{WriteRead}_r(C_i,SI_i)$\;  \label{line:is}
         \If {$\exists SI$ such that ($\forall (C_j,SI_j)\in S:
           C_j=SI)$  \textbf{or}  ($\exists (C_j,SI)\in S: SI[i]=C_i[i]$)}   
              { $SI_i := SI$\;
                output $SI$\label{line:nextsnapshot}\tcp*[l]{Output the next atomic snapshot} 
              $C_i[i]++$\;}
           $C_i := \max(C_1,\ldots,C_n)$ \tcp*[l]{Adopt the maximal
                counter value for each process $j$} 
        }
}
\caption{The $\textit{IIS}\rightarrow\textit{AS}$ simulation algorithm: code for process $i$.}
\label{alg:fl}
\end{algorithm}


\begin{theorem}
\label{th:fl}
 Algorithm~\ref{alg:fl} provides a simulation of the 
 AS model in the IIS model such that, for each run $E$ in IIS, the
 simulated run $E'$ satisfies (1) $\strcorrect(E)=\correct(E')$ and (2)$\forall
 i\in\strcorrect(E)$: $\part(E,i)=\part(E')$.
%
\end{theorem}
\begin{proof}
Consider any run $E$ of Algorithm~\ref{alg:fl}.
First we observe that all atomic snapshots of the simulated processes output
in $E$ are all related by containment, i.e., for
every two snapshot $U$ and $U'$ output in the algorithm in line~\ref{line:nextsnapshot}, we have
$U\leq U'$ or $U'\leq U$, when the two vectors are compared position-wise. 
Indeed, for every output snapshot $U$, there is a round $r$ and a process
$i$, such that all processes that appear in $i$'s immediate snapshot
in round $r$ have put $U$ as their clock vectors.
Since in the algorithm the clock vector $C_i$ is maintained to have the
maximal value seen so far for every process $j$ and by the containment property of immediate snapshot, every process that
took the immediate snapshot in round $r$ or later will compute a 
clock vector $U'\geq U$.  

Therefore, we order all atomic snapshots output in $E$ based on the
containment order, let $U_1,U_2,\ldots$ be the resulting sequence
(here $U_{\ell}\leq U_{\ell+1}$ for each $\ell=1,2,\ldots$). 
Then for each $\ell=1,2,\ldots$ and for each process $i,\;
U_{\ell+1}[i]\neq U_{\ell}[i]$, we add an update operation in which $i$ increments its
counter (initially $1$) and writes the result to position $i$ in the memory just
before $U_{\ell+1}$. 
Notice that since a process only increments its counter after it has
output a snapshot, $U_{\ell+1}[i]\neq U_{\ell}[i]$ implies that $U_{\ell+1}[i]=U_{\ell}[i]+1$.

We call the resulting sequence $E'$ and observe that
it is a run of the AS model. Indeed,  the snapshots taken in $E'$ are related by
containment and, by construction, each snapshot returns the latest
written value for each process.  By construction, $E$ and $E'$ agree on the sequence of snapshots taken
by every given process. 
Moreover, since the clock vector of process $i$ contains the most
up-to-date value for every other process and in the first step each
process simply writes its initial clock vector in the memory, the set
of participating processes as observed by $i$ in $E$ is the same as
the set of participating processes observed by $i$ in $E'$.
Thus, $E$ is an AS run, and Algorithm~\ref{alg:fl} simulates AS in IIS.

Every update of the counter of a strongly correct process $i$ eventually appears
in the snapshot of every other strongly correct process. Thus, every simulated snapshot of a strongly correct
process eventually completes and $\part(E,i)=\part(E')$. If a process is not
strongly correct, it eventually blocks in trying to complete its snapshot. Thus,
$\strcorrect(E)=\correct(E')$.
\end{proof}

\section{Conclusion}

This paper presents two simulation algorithms that, taken together,
maintain the equality between the set of \emph{correct} processes in
AS and the set of \emph{strongly correct} processes
in IIS. 
This equality enables a strong equivalence relation between AS and IIS
\emph{sub-models}: an object is implementable in an adversarial sub-AS model
if and only if it is implementable in the corresponding adversarial
sub-IIS model. 
The result holds regardless of whether the object is one-shot, like a
distributed task, or long-lived, like a queue or a counter. (Naturally,
in IIS, we guarantee liveness of object operations to the strongly
correct processes only.)   
The equivalence presented in this paper motivates the recent topological characterization of
task computability in sub-IIS models~\cite{GKM14-podc} and suggests
further exploration of iterated models that capture, besides
adversaries~\cite{DFGT11}, the use of generic tasks like the M\"obius
task~\cite{GRH06} 
or of a task from the family of $0$-$1$ exclusion~\cite{Gaf08-01}.


{\small
\bibliography{references}
}


\end{document}